\documentclass[acmsmall]{ec20acm}
\usepackage{booktabs} 
\usepackage[ruled]{algorithm2e} 

\SetAlFnt{\small}
\SetAlCapFnt{\small}
\SetAlCapNameFnt{\small}
\SetAlCapHSkip{0pt}
\IncMargin{-\parindent}

 \usepackage{enumitem}

\usepackage[framemethod=tikz]{mdframed}

\let\emptyset\varnothing

\theoremstyle{plain}
\newtheorem{theorem}{Theorem}
\newtheorem*{theorem*}{Theorem}
\newtheorem{lemma}[theorem]{Lemma}

\newtheorem{corollary}[theorem]{Corollary}
\newtheorem{fact}{Fact}
\newtheorem{informal}{Informal Theorem}

\theoremstyle{definition}
\newtheorem{definition}{Definition}

\theoremstyle{remark}

\newcommand{\poly}{\operatorname{poly}}

\newcommand{\mdmdp}{\textsf{MDMDP}}
\newcommand{\odp}{\textsf{ODP}}
\newcommand{\sadp}{\textsf{SADP}}
\newcommand{\opt}{\textsc{OPT}}


\copyrightyear{2020}
\acmYear{2020}
\setcopyright{acmlicensed}\acmConference[EC '20]{Proceedings of the 21st ACM Conference on Economics and Computation}{July 13--17, 2020}{Virtual Event, Hungary}
\acmBooktitle{Proceedings of the 21st ACM Conference on Economics and Computation (EC '20), July 13--17, 2020, Virtual Event, Hungary}
\acmPrice{15.00}
\acmDOI{10.1145/3391403.3399496}
\acmISBN{978-1-4503-7975-5/20/07}
\title{On the (in)-approximability of Bayesian Revenue Maximization for a Combinatorial Buyer}
\author{Natalie Collina}
\authornote{nataliecollina@gmail.com.}
\affiliation{%
  \institution{Princeton University}
}

\author{S. Matthew Weinberg}
\authornote{smweinberg@princeton.edu. Supported by NSF CCF-1717899.}
\affiliation{%
  \institution{Princeton University}
}

\renewcommand{\shortauthors}{N.~Collina, S.M.~Weinberg.}
\date{}

\begin{abstract}
We consider a revenue-maximizing single seller with $m$ items for sale to a single buyer whose value $v(\cdot)$ for the items is drawn from a known distribution $D$ of support $k$. A series of works by Cai et al. establishes that when each $v(\cdot)$ in the support of $D$ is additive or unit-demand (or $c$-demand), the revenue-optimal auction can be found in $\poly(m,k)$ time.

We show that going barely beyond this, even to matroid-based valuations (a proper subset of Gross Substitutes), results in strong hardness of approximation. Specifically, even on instances with $m$ items and $k \leq m$ valuations in the support of $D$, it is not possible to achieve a $1/m^{1-\varepsilon}$-approximation for any $\varepsilon>0$ to the revenue-optimal mechanism for matroid-based valuations in (randomized) poly-time unless NP $\subseteq$ RP (note that a $1/k$-approximation is trivial). 

Cai et al.'s main technical contribution is a black-box reduction from revenue maximization for valuations in class $\mathcal{V}$ to optimizing the difference between two values in class $\mathcal{V}$. Our main technical contribution is a black-box reduction in the other direction (for a wide class of valuation classes), establishing that their reduction is essentially tight.
\end{abstract}
\begin{document}
\maketitle

\section{Introduction}\label{sec:intro}
``Multi-dimensional mechanism design'' has been a driving force in Mathematical Economics since Myerson's seminal work~\cite{Myerson81}, and also in Algorithmic Game Theory since its introduction to TCS by seminal work of Chawla, Hartline, and Kleinberg~\cite{ChawlaHK07}. The problem has gained broad interest within TCS owing to the complexity of optimal solutions~\cite{Thanassoulis04, Pavlov11,BriestCKW10, HartN13, DaskalakisDT15, DaskalakisDT17, HartR15, RubinsteinW15, DaskalakisDT14, ChenDPSY14, ChenDOPSY15}. In light of this, there now exists a substantial body of work providing algorithms to find optimal (or approximately optimal) auctions, accepting that the resulting solution may be randomized, non-monotone, or not particularly simple~\cite{Alaei11, CaiDW12a, CaiDW12b, AlaeiFHHM12, AlaeiFHH13, CaiDW13a, CaiDW13b, DaskalakisW15, DaskalakisDW15}.

To properly parse the results of our work, we need to be clear about the input model and objective. The seller has $m$ items for sale to a single buyer, and is given as input an explicit finite-support distribution $D$ over valuation functions in some class $\mathcal{V}$ (given by listing the valuations along with the probability with which they are drawn from $D$).\footnote{We will formally specify in Section~\ref{sec:prelim} how valuation functions are presented as input, but quickly note that all results in this line of work (including ours) are compatible with any standard input format, such as value oracles, demand oracles, or an explicit poly-sized circuit which computes either query.} The designer's goal is to output a description of an auction which (approximately) maximizes expected revenue. The main positive results of these prior works provide poly-time algorithms for the revenue-optimal \emph{multi-buyer} auction, when all buyers are additive/unit-demand/$c$-demand (which of course imply poly-time algorithms for a single buyer as well, although the single-buyer case does not require many of the developed tools).\footnote{A valuation is $c$-demand if $v(S):= \max_{T \subseteq S, |T| \leq c}\{\sum_{i \in T} v(\{i\})\}$. Additive valuations are $m$-demand, and Unit-demand valuations are $1$-demand.}~\cite{CaiDW13b} also establishes a hardness result: there exists a constant $d$ such that it is not possible to guarantee a $1/m^d$-approximation for a single buyer whose valuation is submodular\footnote{A function $v(\cdot)$ is submodular if $v(S \cap T) + v(S \cup T) \leq v(S) + v(T)$ for all $S, T$.} in poly-time unless NP $\subseteq$ RP. Our first main result establishes that revenue-maximization is inapproximable immediately beyond $c$-demand valuations:

\begin{informal}[See Theorem~\ref{thm:matroidbased}] Unless NP $\subseteq$ RP, for all $\varepsilon>0$, there is no poly-time $1/m^{1-\varepsilon}$ approximation to the revenue-optimal auction for a single buyer whose values for $m$ items are drawn from a distribution of support $\leq m$ over \emph{matroid-based} valuations.\footnote{A matroid-based valuation satisfies $v(S):= \max_{T \subseteq S, T \in \mathcal{I}}\{\sum_{i \in T} v(\{i\})\}$, where $\mathcal{I}$ is a matroid.}
\end{informal}

Importantly, note that a $1/\text{support}(D)$-approximation is trivial, so Theorem~\ref{thm:matroidbased} rules out essentially any non-trivial approximation. To best understand our proof of Theorem~\ref{thm:matroidbased}, we must first overview the proof approach of~\cite{CaiDW13b} for their positive results.~\cite{CaiDW13b} establishes an approximation-preserving black-box reduction from revenue-maximization when all valuations come from class $\mathcal{V}$ (refer to this problem as \mdmdp($\mathcal{V}$), formal definition in Section~\ref{sec:prelim}) to maximizing the difference of two functions in $\mathcal{V}$ (refer to this problem as \odp($\mathcal{V}$), formal definition in Section~\ref{sec:prelim}). For simple classes like additive valuations, \odp\ is easy, and this yields their positive results. For complex valuation classes like submodular,~\cite{CaiDW13b} establish that \odp\ is hard, and this intuition drives their inapproximability result. 

However,~\cite{CaiDW13b} lacks a formal reduction from \odp\ to \mdmdp. Using intuition that \odp(submodular) is hard, they directly construct hard \mdmdp\ instances. Our next main result is a reduction from \odp\ to \mdmdp. That is, the~\cite{CaiDW13b} reduction is essentially tight.

\begin{informal}[See Theorem~\ref{thm:mainreduction}] For suitable $\mathcal{V}$, there is a black-box reduction from \odp($\mathcal{V}$) to \mdmdp($\mathcal{V}$) which is almost approximation-preserving.
\end{informal}

Theorem~\ref{thm:mainreduction} allows us to reason exclusively about hardness of \odp, and conclude hardness for \mdmdp. Our final result provides a clean framework to quickly establish when \odp($\mathcal{V}$) is inapproximable, and proves that \odp(matroid-based) is indeed inapproximable (and also that matroid-based is suitable for the reduction in Theorem~\ref{thm:mainreduction}). 

\subsection{Context and Related Work}\label{sec:related}
The context in which to view our work is the following:~\cite{CaiDW13b} establishes a reduction from \mdmdp\ to \odp, and our Theorem~\ref{thm:mainreduction} establishes a reduction from \odp\ to \mdmdp, so their reduction is tight. Moreover, \odp\ is now a simple lens through which one can study (tightly!) the computational complexity of mechanism design --- revenue-maximizing Bayesian mechanism design for valuations in class $\mathcal{V}$ is \emph{exactly} as hard as maximizing the difference of two functions in $\mathcal{V}$ (\cite{CaiDW13b} provides one direction of the reduction, and we provide the other).

Two previous works prove hardness of revenue maximization in this model. As already discussed,~\cite{CaiDW13b} uses intuition from hardness of \odp(submodular) to directly construct hard instances for \mdmdp(submodular). They even provide a partial framework for proving similar results (which we overview in Section~\ref{sec:prelim}). Our results improve these by replacing submodular with matroid-based (a vastly more restrictive class --- it is a proper subset of gross substitutes), and by providing a true reduction from \odp\ to \mdmdp.

The other known hardness result is from~\cite{DobzinskiFK11}, who establish that \mdmdp(OXS)\footnote{OXS is defined in Section~\ref{sec:prelim}.} is NP-hard to solve \emph{exactly}. Their proof does follow a principled framework, and they establish roughly that whenever \odp($\mathcal{V}$) is hard to solve exactly, that \mdmdp($\mathcal{V}$) is hard to solve exactly as well. Their framework, however, is very clearly limited to exact hardness.\footnote{This is because their produced \mdmdp\ instances always have support $2$, and it is trivial to get a $1/2$-approximation on instances of support $2$.} Our results improve  these by providing an approximation-preserving framework, which allows for hardness of approximation. 

There is also a long series of related work in the \emph{independent items} model~\cite{ChawlaHK07, ChawlaHMS10, ChawlaMS15, HartN12, LiY13, BabaioffILW14,RubinsteinW15, Yao15, CaiDW16, ChawlaM16, CaiZ17}. Here, the distribution $D$ is not given explicitly by listing its support, but some form of sample access (or other concise description) is given instead. The restriction is that $D$ satisfies ``independent items''. We refer the reader to~\cite{RubinsteinW15} for the general definition, but the example to have in mind is an additive buyer, where independent items simply means that $v(\{i\})$ is drawn independently of $v(\{j\})$ for all $i \neq j$. Observe that inputting such a distribution explicitly in our model would require input of size $\text{exp}(m)$. So the way to reconcile works such as~\cite{DaskalakisDT14, ChenDPSY14, ChenDOPSY15} (which prove hardness of exact optimization for a single additive/unit-demand buyer) with folklore LPs (which provide poly-time algorithms for exact optimization for a single additive/unit-demand buyer) is that the hardness results rule out $\poly(m)$-time solutions, whereas the LPs run in time $\text{exp}(m)$ (which is the size of the support of the distribution). Similarly, the way to reconcile our results (which prove strong hardness of approximation for matroid-based) with~\cite{RubinsteinW15,ChawlaM16, CaiZ17} (which provides constant-factor approximations for a single subadditive buyer) is that these constant-factor approximations require the independent items assumption.

\subsection{Summary and Roadmap}
Our main results complete the picture for the~\cite{CaiDW13b} reduction, establishing an approximation-presering reduction from \odp\ to \mdmdp, and moreover that \mdmdp\ is inapproximable within non-trivial factors as soon as we move beyond $c$-demand valuations. Section~\ref{sec:prelim} formally states the problems we study, and recaps prior work~\cite{CaiDW13b} in more detail. Section~\ref{sec:odp} establishes a clean framework to prove when \odp\ is inapproximable. Section~\ref{sec:reduction} provides our reduction and concludes with our main theorem statements. The appendix contains all omitted proofs, along with some examples demonstrating interesting facts about some of our tools along the way.

\section{Preliminaries}\label{sec:prelim}
\subsection{Valuation Functions and Input}\label{sec:valuations}
A \emph{valuation function} takes as input a set $S$ of items and outputs a value $v(S)$. Throughout the paper, $[m]$ will denote the ``base set'' of $m$ items (although our reductions will create additional items). {For ease of notation, and to help the reader parse parameters, we will w.l.o.g. scale all valuation functions so that $v(S) \in \mathbb{N}$ for all $S$ (this is w.l.o.g. when $v(S)$ is rational for all $S$, which is itself w.l.o.g. due to $\varepsilon$-truthful-to-truthful reductions of, e.g.,~\cite[Theorem~5.2]{RubinsteinW15}). Moreover, all valuations considered in this paper are normalized ($v(\emptyset) = 0$), monotone ($v(S \cup T) \geq v(S)$ for all $S,T$), and have no trivial items ($v(S) > 0$ for all $S \neq \emptyset$).\footnote{We will confirm that no trivial items is w.l.o.g. once we define our formal problems.} Finally, when $S$ is a random variable, we will abuse notation and let $v(S):= \mathbb{E}_S[v(S)]$. Below are the two main classes of valuations that our main results reference. Appendix~\ref{app:valuations} contains definitions of related classes of interest. Note that OXS $\subsetneq$ Matroid-based $\subsetneq$ Gross Substitutes $\subsetneq$ Submodular.
\begin{itemize}[noitemsep]
\item OXS: there is a weighted bipartite graph $G$ with nodes $L= [m]$ on the left and $R$ on the right. $v(S)$ is the size of the max weight matching using nodes $S$ on the left and $R$ on the right. When all weights are $1$ or $0$, call this binary OXS.
\item Matroid-based: let $\mathcal{I}$ be independent sets of a matroid on $[m]$, and $w_i$ be weights for each $i\in [m]$. Then $v(S):=\max_{T \subseteq S, T \in \mathcal{I}}\{\sum_{i \in T} w_i\}$. When each $w_i = 1$, call this matroid-rank.
\end{itemize}

\paragraph{Representing Valuation Functions.} All problems we consider require a valuation function to be ``input.'' Our theorem statements will be precise about what input models are assumed, although our results hold for most reasonable input models not explicitly discussed.
\begin{itemize}
\item Value oracle: each $v(\cdot)$ is given via a black box which takes as input a set $S$ and outputs $v(S)$. 
\item Demand oracle: each $v(\cdot)$ is given via a black box which takes as input a vector $\vec{p}$ of prices and outputs a set in $\arg\max_T\{v(T) - \sum_{i \in T}p_i\}$.
\item Explicit input: each $v(\cdot)$ is given explicitly via a circuit or Turing machine which takes as input a set $S$ and outputs $v(S)$. We will only consider succinct representations (that is, circuits/Turing machines which are of polynomial size/runtime for inputs of size $m$). 
\end{itemize}

\subsection{Formal Problem Statements}\label{sec:problems}
The first problem we study is simply Bayesian revenue maximization, but we will be precise with how the input is specified to correctly place it with prior work. We use the language of~\cite{CaiDW13b} when possible to draw connections, although we will drop unnecessary parameters.\\

\noindent \mdmdp($\mathcal{V}$) --- Multi-Dimensional Mechanism Design Problem for class $\mathcal{V}$:\\
\textsc{Input}: an explicit distribution $D$ over $k$ valuation functions in $\mathcal{V}$ (given by listing all $v(\cdot)$ in the support, and the probability with which they are drawn from $D$).\\
\textsc{Output}: for each $v(\cdot)$ in the support of $D$, a (possibly randomized) set $S_v$ and a price $p_v$ such that $v(S_v) - p_v \geq v(S_w)-p_w$ for all $v,w$ in the support of $D$.\\
\textsc{Objective}: maximize $\sum_v \Pr[v\leftarrow D] \cdot p_v$, the expected revenue.\\
\textsc{Approximation}: a solution guarantees an $\alpha$-approximation if $\sum_v \Pr[v \leftarrow D] \cdot p_v \geq \alpha \cdot \opt$.\\

The lens by which we study \mdmdp\ is the following standard optimization problem. \\

\noindent \odp($\mathcal{V}$) --- Optimize Difference Problem for class $\mathcal{V}$:\\
\textsc{Input}: two functions $v(\cdot)$ and $w(\cdot)$, both in $\mathcal{V}$.\\
\textsc{Output}: a (possibly randomized) set, $S$.\\
\textsc{Objective}: maximize $v(S) - w(S)$.\\
\textsc{Approximation}: a solution guarantees an $\alpha$-approximation if $v(S) - w(S) \geq \alpha \cdot \opt$.\footnote{Observe above that if both $v(\cdot)$ and $w(\cdot)$ are subadditive, and $v(T) = 0$, then $v(S \cup T) = v(S)$ for all $S$, while $w(S \cup T) \geq w(S)$ for all $S$. Therefore, it is without loss to remove the items in $T$ from consideration. Similarly, if $w(T) = 0$, then $w(S \cup T) = w(S)$ for all $S$, while $v(S \cup T) \geq v(S)$ for all $S$. Therefore, it is without loss to solve \odp\ after removing $T$, and then add all items in $T$ back at the end. We therefore will pre-process any input to \odp\ by first finding all items $i$ such that $v(\{i\}) = 0$ and removing them, and all items $i$ such that $w(\{i\}) = 0$ and removing them (to add back later). Therefore, it is indeed w.l.o.g. to assume no trivial sets.}\footnote{We will also assume w.l.o.g. that there exists an $S$ for which $v(S) > w(S)$. This is without loss because if we have an algorithm $A$ which succeeds only on such instances, we can first run this algorithm on an arbitrary instance to get a set $T$, and check if $v(T) \geq w(T)$. If so, then output this (and either we were in a case where the algorithm succeeds, or the optimum is $0$ and the algorithm succeeds anyway). If not, then output $\emptyset$ (because $v(T) < w(T)$ would prove that our algorithm failed, and therefore we are in a case where the optimum must be $0$). So we will also assume this w.l.o.g. for our future reductions.}\\

The final problem we study was introduced by~\cite{CaiDW13b}, and shown to have connections to \mdmdp. \sadp\ essentially provides a list of related \odp s, and allows a solution to any of them.\\

\noindent \sadp($\mathcal{V}$) --- Solve Any Differences Problem for class $\mathcal{V}$:\\
\textsc{Input}: a finite list of functions $v_1(\cdot),\ldots, v_k(\cdot)$, all in $\mathcal{V}$.\\
\textsc{Output}: a (possibly randomized) set, $S$.\\
\textsc{Objective}: for some $j \in [k-1]$, have $S$ maximize $v_j(S) - v_{j+1}(S)$.\\
\textsc{Approximation}: a solution guarantees an $\alpha$-approximation if there exists a $j \in [k-1]$ such that $v_j(S) - v_{j+1}(S) \geq \alpha \cdot \max_T\{v_j(T) - v_{j+1}(T)\}$.

\subsection{Recap of~\cite{CaiDW13b}}\label{sec:CDW}
Finally, we briefly recap the main tools/results from~\cite{CaiDW13b} which are relevant for this paper. Recall that the focus of this paper is on hardness, and we establish our results with just a single buyer. Therefore, we will not recap the results of~\cite{CaiDW13b} in their full multi-buyer generality, but just focus on the single-buyer implications. Below, for any valuation class $\mathcal{V}$, $\mathcal{V}^*$ denotes its \emph{conic closure}. That is, $\mathcal{V}^*$ denotes the closure of $\mathcal{V}$ under non-negative linear combinations. Many natural valuation classes (e.g. submodular, XOS, subadditive, additive) are already closed under conic combinations, but others (unit-demand, $c$-demand, OXS, matroid-based, gross substitutes) are not. 

\begin{theorem}[\cite{CaiDW13b}]\label{thm:CDW1} For all $\mathcal{V}$, there is a poly-time, approximation-preserving black-box reduction from \mdmdp($\mathcal{V}$) to \odp($\mathcal{V}^*$). That is, an $\alpha$-approximation algorithm for $\mdmdp(\mathcal{V})$ (in any input model) exists using $\poly(\text{support}(D),m)$ black-box queries to an $\alpha$-approximation algorithm for \odp($\mathcal{V}^*$) (in that same input model), and additional runtime $\poly(\text{support}(D),m)$.
\end{theorem}

For a single buyer, the positive applications of Theorem~\ref{thm:CDW1} are not particularly impressive, and imply only that \mdmdp\ can be solved exactly whenever buyers are $c$-demand (which could alternatively be deduced by a simple linear program) --- their main positive results are for multiple buyers. Again, our main result is an approximation-preserving reduction in the other direction, from \odp\ to \mdmdp.

\subsubsection{\cite{CaiDW13b}'s Hardness of Approximation}}
As referenced above,~\cite{CaiDW13b} also proves hardness of approximation for \mdmdp(submodular). We make use of their machinery, which requires several definitions to precisely state. Without repeating the entire~\cite{CaiDW13b} hardness of approximation, it is perhaps impossible to motivate why these \emph{precise} definitions are relevant, as they are technical in nature. However, we do give intuition to parse what the definitions are stating, and roughly the role they serve in prior work.

The first definition, Compatibility, is given below. Immediately afterwards, we provide context to help parse the condition.

\begin{definition}[Compatibility] We say that a list of valuation functions $V = (v_1,\ldots, v_k)$ and a list of (possibly randomized) sets $X = (X_1,\ldots, X_k)$ are \emph{compatible} if:
\begin{itemize}
\item $X$ and $V$ are cyclic monotone. That is, the welfare-maximizing matching of valuations $V$ to allocations $X$ (that is, which maximizes $\sum_{i=1}^k v_i(X_{M(i)})$) is to match $X_i$ to $v_i$ for all $i$. 
\item For any $i < j$, the welfare-maximizing matching of valuations $v_{i+1},\ldots, v_j$ to allocations $X_i,\ldots, X_{j-1}$ is to match allocation $X_\ell$ to valuation $v_{\ell+1}$ for all $\ell$.
\end{itemize}
\end{definition}

One should parse Compatibility as a stronger condition than cyclic monotonicity (indeed, the first bullet is precisely cyclic monotonicity). Cyclic monotonicity requires that a particular matching on the complete bipartite graph is max-weight. Compatibility requires that a particular matching on several proper subgraphs are also max-weight. Intuitively, the particular restrictions in bullet two assert that higher-indexed allocations are ``better'', and higher-index valuations are ``more important'' (and therefore, the welfare-maximizing allocation always gives ``better'' allocations to ``more important'' valuations).

The next technical definition is a restriction on potential inputs to \sadp. Intuitively, this condition serves the following purpose: the~\cite{CaiDW13b} reduction from \sadp\ to \mdmdp\ will take an input to \sadp, and pass it on to \mdmdp, and one part of their proof needs to establish the existence of a high-revenue solution to the produced \mdmdp\ instance (so that any $\alpha$-approximation must also produce high revenue). $C$-compatibility suffices for this (but we will not attempt to explain further why this is the case, and refer the reader to~\cite{CaiDW13b} for more detail).

\begin{definition}[$C$-compatible] A list of valuation functions $(v_1,\ldots, v_k)$ is \emph{$C$-compatible} if there exist integers $1 = Q_1 < \ldots < Q_k$, all at most $2^C$, and allocations $(X_1,\ldots, X_k)$ such that:
\begin{itemize}
\item For all $\ell \in [k-1]$, $X_\ell \in \arg\max\{v_\ell(S) - v_{\ell+1}(S)\}$.
\item $(Q_1 \cdot v_1,\ldots, Q_k \cdot v_k)$ is compatible with $(X_1,\ldots, X_k)$. 
\end{itemize}
\end{definition}

The final technical definition in their reduction is balanced-ness. This condition is used in their reduction to guarantee quality of approximation for the original \sadp\ instance. The intuition to have in mind is that when given as input to \mdmdp\ a distribution of support $k$, it is trivial to get a $1/k$-approximation simply by targeting the valuation in the support which maximizes $v([m]) \cdot \Pr[v \leftarrow D]$ (and setting price $v([m])$ for $[m]$, and no other options). Put another way, it's possible to guarantee a $1/k$ (or comparably poor) approximation to \mdmdp\ without engaging with the instance at all. The purpose of their $d$-balanced property guarantees that sufficiently good approximations to the \mdmdp\ instance produced by their reduction must actually engage the initial \sadp\ instance.

\begin{definition}[$d$-balanced] A list of functions $(v_1,\ldots,v_k)$ is $d$-balanced if $v_k([m]) \leq d\cdot (v_\ell(X)-v_{\ell+1}(X))$ for all $\ell \in [k-1]$, $X = \arg\max_T \{v_\ell(T)-v_{\ell+1}(T)\}$.
\end{definition}

Following the intuition of the preceding paragraph, $d$-balanced aims to upper bound the revenue a seller could get on a \sadp\ instance by simply selling $[m]$ to $v_k(\cdot)$ and ignoring everything else. With these two definitions, we may state the reduction of~\cite{CaiDW13b} from \sadp\ to \mdmdp:\footnote{This is indeed a correct statement of~\cite{CaiDW13b} Theorem~7 after chasing through their \sadp\ choice of parameters.}

\begin{theorem}[\cite{CaiDW13b}]\label{thm:CDW2} Let $A$ be an $\alpha$-approximation algorithm for \mdmdp($\mathcal{V}$). Then a solution to any $C$-compatible instance $(v_1,\ldots, v_k)$ of \sadp($\mathcal{V}$) can be found in polynomial time plus one black-box call to $A$. The solution has the following properties:
\begin{itemize}
\item (Quality) If the \sadp\ input is $d$-balanced, then the solution is an $(\alpha - \frac{(1-\alpha)d}{k-1})$-approximation.
\item (Complexity) If for all $i$, $v_i([m])\leq 2^b$ then $w([m])\leq 2^{b+C}$ for all $w(\cdot)$ input to $A$. Moreover, all probabilities input to $A$ can be written as the ratio of two integers at most $2^{2C}$. 
\end{itemize}
\end{theorem}

The intended application of Theorem~\ref{thm:CDW2} is then to find a hard instance of \sadp\ which is (a) $(\ll k)$-balanced --- this guarantees that the resulting guarantee on \sadp\ is close to $\alpha$, and (b) $\poly(m)$-compatible --- this guarantees that the input passed on to \mdmdp\ blows up by only a $\poly(m)$ factor. The appealing feature of Theorem~\ref{thm:CDW2} is that it suffices to establish that \mdmdp(submodular functions) is inapproximable within any polynomial factor, by directly constructing a hard instance of \sadp\ with $k=\poly(m)$ valuations which is $(\ll k)$-balanced and $\poly(m)$-compatible. 

The unappealing feature of Theorem~\ref{thm:CDW2} is that \sadp\ is not a particularly natural problem to think about, nor are the compatibility/balanced properties (hence, the need for a substantial preliminary section just to state their result). Additionally, \odp\ is a very special case of \sadp\, whose solution suffice for \mdmdp\, but Theorem~\ref{thm:CDW2} only establishes that \mdmdp\ is hard when \sadp\ is hard for large values of $k$. 

Our work addresses both shortcomings. Theorem~\ref{thm:mainreduction} provides a formal statement, which essentially replaces \sadp\ with \odp\ (truly establishing that the~\cite{CaiDW13b} reduction is tight), and removes any compatibility/balanced requirements on the input instance (allowing greater ease of application to valuation classes significantly more restrictive than submodular).
\section{When is \odp\ Hard?}\label{sec:odp}
Before diving into the technical part of our reduction, we'd first like a clean way to reason about valuation classes for which \odp\ is hard. Becuase \odp\ has a mixed-sign objective, one naturally expects that it is either solvable exactly in poly-time, or hard to even distinguish whether the optimum is non-zero (although there are sometimes exceptions to this intuition~\cite{DaskalakisDW15}). What's not immediately clear is how rich $\mathcal{V}$ needs to be before \odp\ becomes unsolvable in poly-time (e.g. it is solvable for additive functions, but not for submodular, what about in between?). In this section, we show that even \odp(binary OXS) is inapproximable. We begin with a general construction of hard instances:

\begin{definition}[Perturbable] For $v, S$, define $v_S(T):= v(T)$ if $S \neq T$, and $v_S(S) := v(S)-1$. Class $\mathcal{V}$ is \emph{$(x,y)$-perturbable} if there exists a $v(\cdot) \in \mathcal{V}$ such that $v_S(\cdot) \in \mathcal{V}$ for $x$ distinct $S$, and $v([m]) = y$.

$\mathcal{V}$ is efficiently $(x,y)$-perturbable if further there is a computationally-efficient bijection from the $x$ sets to $[x]$, and there is a poly-sized circuit/poly-time Turing machine which computes $v(\cdot)$.
\end{definition}

The parameter $y$ will not be used in this section, and is not necessary to establish when \odp\ is hard, but we will need to reference the parameter in the technical parts of our reduction. Lemma~\ref{lem:perturb} establishes that \odp\ is hard for highly-perturbable ($(x,y)$-perturbable for large $x$) classes. Intuitively, this is because the $x$ \odp\ instances of the form $(v,v_S)$ each have disjoint solutions, but are hard to distinguish. Proofs of the following two lemmas appear in Appendix~\ref{app:odp}.

\begin{lemma}\label{lem:perturb} Let $\mathcal{V}$ be $(x,y)$-perturbable. Then any value- or demand-oracle algorithm for \odp($\mathcal{V}$) which guarantees an $\alpha$-approximation \emph{for any $\alpha > 0$} with probability $q$ makes $\geq \frac{qx-1}{2}$ queries.

If $\mathcal{V}$ is efficiently $(x,y)$-perturbable for $x = \text{exp}(m)$, then no randomized, poly-time algorithm for explicit input guarantees an $\alpha$-approximation for any $\alpha > 0$ w.p. $1/\poly(m)$, unless NP $\subseteq$ RP.
\end{lemma}

\begin{lemma}\label{lem:boxs} Binary OXS is efficiently $(\binom{m}{m/2},m/2)$-perturbable.\footnote{So every superclass of Binary OXS (including Matroid-based) is efficiently $(\binom{m}{m/2},m/2)$-perturbable as well.} 
\end{lemma}

Lemmas~\ref{lem:perturb} and~\ref{lem:boxs} together immediately conclude that \odp(Binary OXS) is inapproximable within any non-zero factor with $\poly(m)$ value or demand queries, or within poly-time unless NP $\subseteq$ RP.
\section{From \odp\ to \sadp}\label{sec:reduction}
Now that we know \odp\ is hard even for fairly basic combinatorial valuations, we wish to leverage this to establish that the particular \sadp\ instances needed for Theorem~\ref{thm:CDW2} are also hard. Again, the goal of this section is to provide a reduction from \odp\ to \sadp, while also chasing through the parameters it implies for a full reduction from \odp\ to \mdmdp. Before providing our reduction, let's get some intuition for the main challenges through some simple failed attempts. Recall that an $\alpha$ approximation for \mdmdp\ implies an $(\alpha - \frac{(1-\alpha)d}{k-1})$ approximation for \sadp. We want this to be as close to $\alpha$ as possible, so the properties we care about are:
\begin{itemize}
\item How balanced is the resulting \sadp\ instance? (Smaller $d$ is better).
\item How many functions are input to the \sadp\ instance? (Larger $k$ is better).
\item Is the instance $\poly(m)$-compatible? (Necessary for the reduction to be valid --- otherwise it will produce inputs to \mdmdp\ of super-polynomial size).
\end{itemize}

\vspace{2mm}\noindent\textbf{Failed Attempt 1.} The obvious first attempt is to just observe that \odp\ is a special case of \sadp\ and simply ``reduce'' from \odp\ to \sadp\ by copying the original instance itself into Theorem~\ref{thm:CDW2}. This reduction is clearly correct, but let's examine what parameters we get from this reduction, using our $(\binom{m}{m/2},m/2)$-perturbable $v$ as a representative instance:
\begin{itemize}
\item In that construction, $w([m]) = m/2$, and $\max_{T}\{v(T) - w(T)\} = 1$, so it is $(m/2)$-balanced.
\item There are two functions input to \odp, so $k=2$.
\end{itemize}

From above, $d/(k-1) = m/2$. So even if the instance were $\poly(m)$-compatible, an $\alpha$-approximation for \mdmdp\ wouldn't even imply a non-trivial guarantee for any $\alpha \leq \frac{1}{1 + 2/m}$. Recall again the intuition for this: in the~\cite{CaiDW13b} reduction, $k$ will be the size of the support of the produced $d$ for \mdmdp. If $k=2$, it's trivial to get a $1/2$-approximation without engaging \sadp\ at all. So in order to get stronger hardness of approximation, we need much larger $k$.

\vspace{2mm}\noindent\textbf{Failed Attempt 2.} Here is a next attempt, which takes $k$ as large as desired. Let $v,w$ be the given \odp\ instance, and define $v_\ell(\cdot):= (k-\ell)\cdot v(\cdot) + (\ell-1)\cdot w(\cdot)$. Observe that for any $\ell$, $v_\ell(\cdot)-v_{\ell+1}(\cdot) = v(\cdot)-w(\cdot)$, so \emph{every} consecutive difference is solving exactly the \odp\ instance we care about, and an $\alpha$-approximation for this \sadp($v_1,\ldots, v_k$) will indeed yield an $\alpha$-approximation for \odp($v,w$). Let's check the parameters when we apply Theorem~\ref{thm:CDW2} with our hard \odp\ instance:
\begin{itemize}
\item $w([m]) = m/2$, and our $v_k([m]) = (k-1)\cdot w([m]) = (k-1)m/2$. Also, $\max_T\{v(T)-w(T)\} = 1$, so $\max_T\{v_\ell(T) -v_{\ell+1}(T)\} = 1$ for all $\ell$, and the instance is $(k-1)m/2$-balanced.
\item We can set $k$ as we please.
\end{itemize}

Again, even if the instance were $\poly(m)$-compatible, an $\alpha$-approximation for \mdmdp\ still wouldn't imply non-trivial guarantees for any $\alpha \leq \frac{1}{1 + 2/m}$ (because still $d/(k-1) = m/2$). What we learn from these first two failed attempts is that we need to somehow let $k$ be large in our construction (have many functions), while simultaneously not letting $v_k([m])$ grow too big.

\vspace{2mm}\noindent\textbf{Failed Attempt 3.} Our final failed attempt will also demonstrate one of our key insights: introduce extra items in the \sadp\ instance. Specifically (for this construction), if we are initially given an \odp\ instance on $m$ items, and wish to produce a \sadp\ instance with $k$ valuation functions, add an additional $k$ items. Call the initial items $M$ and the additional items $K$. For any set $S \subseteq M \cup K$, write it as $S=S_M \sqcup S_K$ (where $S_M = S \cap M, S_K = S \cap K$).\footnote{As $S_M \cap S_K = \emptyset$, $S_M \sqcup S_K = S_M \cup S_K$. We use the disjoint union notation here (and in several other places) to remind the reader that these sets are always disjoint.} Define $v_\ell(S)$ to be equal to $v(S_M)$ if and only if $|S_K| \geq \ell$, and $w(S_M)$ if and only if $|S_K| \leq \ell-1$ (note that this is a somewhat bizarre valuation function, which essentially treats the items in $K$ as a bit to decide whether the items in $M$ are valued using $v(\cdot)$ or $w(\cdot)$). It is not hard to see that any $\alpha$-approximation on this \sadp\ instance implies an $\alpha$-approximation on our desired \odp\ instance. Let's again compute the parameters of the reduction for our hard \odp\ instance.

\begin{itemize}
\item $w([m]) = m/2$, and our $v_k([m]) = w([m]) = m/2$. Also, it is not hard to see that for a given $\ell$, $\max_T\{v_\ell(T) - v_{\ell+1}(T)\} = 1$. So the instance is $(m/2)$-balanced.
\item We can set $k$ as we please. 
\end{itemize}

Now, we're in good shape with parameters: we could (for example) set $k = m^c$ for any constant $c$, and establish that any $m^{-c+1}$-approximation for \mdmdp\ instances with $m^c$ items implies a $(>0)$-approximation for \odp\ on $m$ items, which we ruled out in Section~\ref{sec:odp}. The catch is that our produced \sadp\ instance has \emph{bizarre} valuation functions. So all we have done is shown that \mdmdp(bizarre valuation class) is hard. 

\subsection{Our Reduction}
The previous examples illustrate the reasoning one needs to go through to analyze a reduction, and also identify some of the technical challenges. The two main (good) ideas we saw from the failed attempts were (a) introducing extra items and (b) ``truncating'' the valuations in a way that let $k$ grow large while keeping $d$ small. The challenge we did not yet overcome was how to do this in a way that preserves membership in interesting valuation classes, and we did not attempt to address Compatibility. We'll present our construction in two steps. The first adds additional items, but does not yet truncate.\\

\noindent\textbf{Reduction Step One: Scaled Disjoint Unions.}\\
\textsc{Input}: $v(\cdot), w(\cdot)$, both in $\mathcal{V}$ (input to \odp($\mathcal{V}$)), $k$ (desired number of valuations for \sadp).\\
\textsc{Output}: $v_1(\cdot),\ldots, v_k(\cdot)$. Each $v_\ell(\cdot)$ takes as input subsets of $(k-1)m$ items. Partition these items into $M_1\sqcup\ldots\sqcup M_{k-1}$, with each $|M_\ell| = m$. For any $S \subseteq M_1 \sqcup\ldots\sqcup M_{k-1}$, write it as $S:= S_1 \sqcup \ldots\sqcup S_{k-1}$ (with each $S_\ell:=S \cap M_\ell$). Define $v_\ell(S):= \sum_{i \geq \ell} (k+i) \cdot v(S_i) + \sum_{i < \ell} (k+i) \cdot w(S_i)$.\\

That is, we make $k-1$ copies of the $m$ items, have each $v_\ell(\cdot)$ behave either as $v$ or $w$ (scaled) on each copy, and then sum them. Observe that $v_\ell(\cdot)$ switches from $v$ to $w$ exactly at copy $\ell$, and that each $v_\ell(\cdot)$ is a \emph{disjoint union} of valuations in $\mathcal{V}$, formally defined below. This operation happens to preserve membership in the valuation classes we care about. 

For scaling, observe that all classes we consider except for Binary OXS and matroid-rank (because they insist on binary values) are closed under non-negative scaling. The scaling in each partition will not become relevant until the very end when we need to prove that the produced instance is $\poly(m)$-compatible. 

\begin{definition}[Disjoint Union of Valuations] The \emph{disjoint union} of two valuations $v(\cdot)$ and $w(\cdot)$ each on items $M$ creates two copies $M_1 \sqcup M_2$, and outputs a valuation $z(\cdot)$ such that $z(S):= v(S \cap M_1) + w(S \cap M_2)$. 
\end{definition}

\begin{fact}Additive, Binary OXS, OXS, Matroid-rank, Matroid-Based, Gross Substitutes, Submodular, XOS, and Subadditive valuations are closed under Disjoint Union. Unit-demand/$c$-demand are not. 
\end{fact}

So we're in good shape so far in the sense that if we start with (say) matroid-based valuations in our \odp\ instance, we get back matroid-based valuations in our \sadp\ instance. But we have not yet made any progress towards bounding the balanced-ness. Our next step is to properly truncate the valuations to make the valuations more balanced. We define two truncations below, which may inspire different related work. However, only Item Truncation is necessary for our main result, so we will focus the remaining presentation and formal statements on Item Truncation.\\

\noindent\textbf{Reduction Step Two: Item Truncation.}\\
\textsc{Input}: $v_\ell(\cdot)$, a valuation function and $y$, a desired truncation.\\
\textsc{Output}: $v'_\ell(\cdot)$, defined so that $v'_\ell(S):= \max_{T \subseteq S, |T| \leq y} \{v_\ell(T)\}$. \\

\noindent\textbf{Alternate Reduction Step Two: Value Truncation.}\\
\textsc{Input}: $v_\ell(\cdot)$, a valuation function and $x$, a desired truncation.\\
\textsc{Output}: $v'_\ell(\cdot)$, defined so that $v'_\ell(S):= \min\{v_\ell(S),x\}$.\\

We prove in Appendix~\ref{app:reduction} that matroid-based valuations are closed under item truncation, as this is the class/truncation we'll use for our main results. We state afterwards which other classes are closed under each truncation (proofs omitted, as we will not use these results elsewhere).

\begin{lemma}\label{lem:matroid} The class of matroid-rank valuations is closed under item truncation for any $y$.
\end{lemma}

\begin{fact}\label{fact:value} Unit-demand, Submodular, XOS, and Subadditive are closed under value truncation for all $x$. Additive, $k$-demand, Binary OXS, OXS, Matroid-rank, Matroid-based, Gross Substitutes are not (there exist $x$, and a function $v(\cdot)$ in each of these classes for which value-truncating $v(\cdot)$ at $x$ is no longer in that class).
\end{fact}

\begin{fact}\label{fact:item} Unit-demand, Matroid-rank, Gross Substitutes, XOS, and Subadditive are closed under item truncation for all $y$.\footnote{Of these, the claim for Gross Substitutes (like most properties of Gross Substitutes) is non-trivial to establish, and uses the fact that Gross Substitutes functions are \emph{well-layered}~\cite{Leme17}.} Additive, $k$-demand, Binary OXS, OXS, and Submodular are not. 
\end{fact}

Again, recall that Lemma~\ref{lem:matroid} is all that's necessary for our main results (as hardness for matroid-rank implies hardness for all superclasses). Facts~\ref{fact:value} and~\ref{fact:item} are included for the sake of better understanding the truncation operations. One curious conclusion of Fact~\ref{fact:item} is that the property of being closed under item-truncation ``skips'' submodular functions. Note also that value-truncation preserves submodularity but not Gross Substiutes, whereas item-truncation preserves Gross Substitutes but not submodularity.
\subsection{Correctness, Balancedness, and Compatibility}

Now, we want to establish that our reduction correctly solves the initial \odp\ instance, and also examine the parameters of the produced \sadp\ instance. Our full reduction takes an \odp\ instance and first puts it through the Scaled Disjoint Union reduction with parameter $k$ to get $k$ valuations each on $(k-1)m$ items. Then, it puts each of these valuations through item truncation with parameter $m$. For a given \odp\ instance $(v,w)$, call the resulting \sadp\ instance $IT(v,w)$. Appendix~\ref{app:reduction} contains a proof of Lemma~\ref{lem:correctnessitem}, and Appendix~\ref{app:value} contains analogous claims for value truncation. 

\begin{lemma}\label{lem:correctnessitem} Let $S$ be an $\alpha$-approximation to the \sadp\ instance $IT(v,w)$. Then with an additional $O(k^2m)$ runtime/value queries (to the item-truncated values), an $\alpha$-approximation to the \odp\ instance $(v,w)$ can be found.
\end{lemma}

Next, we want to see how well-balanced the instance is. Importantly, observe that the bound in Lemma~\ref{lem:balance} does not grow with $k$.

\begin{lemma}\label{lem:balance} For any \odp\ instance $(v,w)$, the instance $IT(v,w)$ is $2mw([m])$-balanced. 
\end{lemma}

So at this point we have a reduction from \odp\ to very balanced instances of \sadp, which also preserve valuation class membership for matroid-based valuation functions. Our last steps are to ensure that we can indeed compute value queries for our generated functions, and also to ensure that the resulting instance is $\poly(m)$-compatible. Lemma~\ref{lem:valuequeries} below is necessary for the following reason: Lemma~\ref{lem:correctnessitem} requires value queries \emph{on values generated for $IT(v,w)$}. However, our input to \odp\ only provides access to $v,w$, which may not suffice for this. 

\begin{definition} Say that a list of valuation functions $V$ is \emph{query-sufficient} for a list $W$ if (a) it is possible to execute value queries for every function in $W$ using polynomially-many value queries for functions in $V$ and (b) it is possible to write succinct circuits/Turing machines to compute value queries for every function in $W$ using succinct circuits/Turing machines computing value queries for every function in $V$.
\end{definition}

Matroid-based $v,w$, due to connections with greedy algorithms, are indeed query-sufficient.

\begin{lemma}\label{lem:valuequeries} 
When $v,w$ are matroid-based, $v,w$ are query-sufficient for $IT(v,w)$.
\end{lemma}

Lemma~\ref{lem:valuequeries} confirms that if our \odp\ instance is given via explicit input or value queries, then this is enough to simulate value queries on each $v_\ell(\cdot)$ we construct. The last step is compatibility. We'll begin with a general sufficient condition for an instance to be Compatible, and finally establish that our instances satisfy this condition.

\begin{lemma}\label{lem:compatible} Let $v_1,\ldots, v_k$ be such that there exists allocations $X_1,\ldots, X_k$ satisfying:
\begin{enumerate}
\item For all $\ell$, $X_\ell \in \arg\max_S\{v_\ell(S) - v_{\ell+1}(S)\}$.
\item $v_\ell(X_i) \leq C_1$ for all $\ell,i$.
\item $v_\ell(X_\ell) \geq v_\ell(X_{\ell-1}) + 1 \geq v_\ell(X_i)+2$, for all $\ell > 1, i \leq \ell-2$.
\end{enumerate}
Then $(v_1,\ldots, v_k)$ is $k\log_2(C_1)$-Compatible.
\end{lemma}

Our last step is to establish that our $IT(v,w)$ instance satisfies the hypotheses of Lemma~\ref{lem:compatible}.

\begin{corollary}\label{cor:fulldeal} $IT(v,w)$ is $k\log_2(2k\max\{v([m]),w([m])\})$-Compatible.
\end{corollary}

\subsection{Putting Everything Together}
With all the pieces now in place, we can now prove Theorem~\ref{thm:mainreduction}. We briefly note that Theorem~\ref{thm:mainreduction} also holds after replacing item truncation with value truncation (and the proof simply replaces Lemma~\ref{lem:correctnessitem} with Lemma~\ref{lem:correctnessvalues}).

\begin{theorem}\label{thm:mainreduction} Let $\mathcal{V}$ be closed under non-negative scaling, disjoint union, and item truncation, and let $(v,w)$ be query-sufficient for $IT(v,w)$ whenever $v,w \in \mathcal{V}$. Let also $A$ be an $\alpha$-approximation algorithm for \mdmdp($\mathcal{V}$) using either value oracles or explicit input. Then for any integer $k \geq 2$, given black-box access $A$, there is a $\poly(km)$-time algorithm for \odp($\mathcal{V}$) making a single black-box call to $A$ plus an additional $\poly(km)$ value queries to the \odp\ input with the following properties:
\begin{itemize}
\item (Quality) On \odp\ input $(v,w)$, the algorithm produces an $(\alpha - \frac{(1-\alpha)2mw([m])}{k-1})$-approximation.
\item (Complexity) Each value $v_\ell(\cdot)$ input to $A$ takes as input subsets of $[(k-1)m]$, and satisfies $v_\ell([(k-1)m]) \leq (2k\max\{v([m]),w([m])\})^k$. Moreover, all probabilities input to $A$ can be written as the ratio of two integers at most $(2k\max\{v([m]),w([m])\})^{2k}$.
\end{itemize}
\end{theorem}
\begin{proof}
To complete the proof, we just need to carefully track all the pieces we put together. The fact that $\mathcal{V}$ is closed under non-negative scaling, disjoint union, and item truncation guarantees that $IT(v,w)$ will only contain valuations in $\mathcal{V}$ (they will also take as input subsets of $[(k-1)m]$). By Lemma~\ref{lem:correctnessitem}, any $\alpha$-approximation to the \sadp\ instance $IT(v,w)$ yields an $\alpha$-approximation to the \odp\ instance $(v,w)$ (in poly-time, because $(v,w)$ is query-sufficient for $IT(v,w)$). So our goal is to solve the \sadp\ instance $IT(v,w)$, which is an instance of \sadp($\mathcal{V}$). 

Lemma~\ref{lem:balance} establishes that $IT(v,w)$ is $2mw([m])$-balanced, and Corollary~\ref{cor:fulldeal} establishes that $IT(v,w)$ is $k\log_2(2k\max\{v([m]),w([m])\}$-Compatible. Therefore, we can solve the \sadp\ instance $IT(v,w)$ using Theorem~\ref{thm:CDW2}, and the approximation/complexity guarantees match up.
\end{proof}

We can now apply Theorem~\ref{thm:mainreduction} to establish that \mdmdp($\mathcal{V}$) is hard for perturbable $\mathcal{V}$.

\begin{theorem}\label{thm:mainhardness} Let $\mathcal{V}$ be closed under non-negative scaling, disjoint union, and item truncation, and let $(v,w)$ be query-sufficient for $IT(v,w)$ whenever $v,w \in \mathcal{V}$. 
\begin{itemize}
\item For any integer $k \geq 2$, if $\mathcal{V}$ is $(x,y)$-perturbable, then no better than a $\frac{2my}{k-1+2my}$-approximation for \mdmdp($\mathcal{V}$) can be guaranteed with probability $q$ on instances with $(k-1)m$ items with fewer than $\frac{qx-1}{2} - \poly(km)$ queries the value oracle model.\footnote{Observe that because of the $-\poly(km)$, this result only has bite when $qx \gg \poly(km)$.}
\item For any $k = \poly(m)$, if $\mathcal{V}$ is efficiently $(x,y)$-perturbable for $x = \text{exp}(m)$, then no better than a $\frac{2my}{k-1+2my}$-approximation for \mdmdp($\mathcal{V}$) can be guaranteed with probability $1/\poly(m)$ on instances with $(k-1)m$ items in the explicit input model, unless NP $\subseteq$ RP.
\end{itemize}
\end{theorem}

Theorem~\ref{thm:mainhardness} implies that \mdmdp(matroid-based) is inapproximable within non-trivial factors.

\begin{theorem}\label{thm:matroidbased} For all $\varepsilon>0$, there does not exist an algorithm making $\poly(m)$ value queries in the value oracle model, $\poly(m)$ demand queries in the demand oracle model, or a poly-time algorithm in the explicit input model (unless NP $\subseteq$ RP) guaranteeing an $1/m^{1-\varepsilon}$-approximation with probability at least $1/\poly(m)$ for \mdmdp(matroid-based) on instances with $m$ items and $m^{1-2\varepsilon/9}$ values in the support of $D$.
\end{theorem}

\bibliographystyle{ACM-Reference-Format}
\bibliography{MasterBib}
\appendix

\section{Omitted Proofs from Section~\ref{sec:odp}}\label{app:odp}

\begin{proof}[Proof of Lemma~\ref{lem:perturb}]
Let $v$ witness that $\mathcal{V}$ is $(x,y)$-perturbable, and consider being given the \odp\ input $(v, v_S)$, for some $S$ such that $v_S \in \mathcal{V}$ chosen uniformly at random. Then $S$ is the unique optimum, and the only set which guarantees an $\alpha$-approximation for any $\alpha > 0$. First consider any (randomized) value-oracle algorithm which makes at most $s$ queries. Using the principle of deferred decisions, first select the (random) prices to query \emph{conditioned on all previous responses being consistent with $v(\cdot)$}. Then as only one possible query ($S$) will be inconsistent with $v(\cdot)$, and $S$ is chosen uniformly at random, we have that except with probability $s/x$, all queries are consistent with $v(\cdot)$. When all queries are consistent with $v(\cdot)$, the algorithm must make some guess, which will be correct with probability at most $1/(x-s)$. So the total success probability is at most $s/x + 1/(x-s)$. Setting $s = \frac{qx-1}{2}$ has this always at most $q$.

Next, consider (randomized) demand-oracle algorithms which make at most $s$ queries. Observe first that the \emph{only} difference between a demand query for $v$ and a demand query for $v_S$ can come on a vector of prices $\vec{p}$ for which the demand query for $v$ on $\vec{p}$ outputs $S$ (this is because if any set $T \neq S$ is $v$'s favorite, it is certainly $v_S$'s favorite. But if $S$ is $v$'s favorite, it may not be $v_S$'s favorite). So again use the principle of deferred decisions, and select the (random) sets to query conditioned on all previous responses being consistent with $v(\cdot)$. Then for each query made, there is at most one $S$ for which $v_S$ might have output something different. Therefore, except with probability $s/x$, all queries are consistent with $v(\cdot)$. Again, the algorithm's total success probability is therefore at most $s/x + 1/(x-s)$. 

To see the final claim in the explicit input model, we develop a reduction from unique-SAT (3-SAT where the promise is that there is exactly one or zero satisfying assignments). Given an instance $X$ of 3-SAT, consider a circuit/Turing machine which takes as input an assignment and outputs yes if that assignment is satisfying (so the circuit would be poly-sized, and the Turing machine would run in poly-time). Define now a circuit/Turing machine for a valuation fuction $v^X(\cdot)$ which takes as input a set $S$ and outputs $v(S)$ if $S$ is not one of the $x$ perturbing sets. Otherwise, first map $S$ to $[x]$, and then take the corresponding assignment $\phi$. If $\phi$ satisfies $X$, output $v(S)-1$. Otherwise, output $v(S)$. Note that this entire procedure is computationally-efficient/can be done with a poly-sized circuit because $\mathcal{V}$ is efficiently perturbable.

Now, assume for contradiction that there is a randomized algorithm for \odp($\mathcal{V}$) which guarantees a non-zero approximation with probably at least $1/\poly(m)$. Given an instance $X$ of unique-SAT, run the algorithm on $v^X(\cdot)$, and let $S$ denote the output set. Observe that by the unique-SAT promise, $v^X(\cdot) \in \mathcal{V}$. Check if the assignment corresponding to $S$ satisfies $X$, and if so, output ``yes'' (otherwise, output ``no''). This algorithm will clearly never mistakenly guess that $X$ is satisfiable. To see that it correctly says yes with probability at least $1/\poly(m)$ (which, as usual, can be amplified to $2/3$ or $1-1/\poly(m)$ if desired), observe that whenever $X$ is satisfiable uniquely by $\phi$, that the set $S$ corresponding to $\phi$ is the \emph{only} set guaranteeing an $\alpha$-approximation for any $\alpha > 0$. Therefore, this set must be output with probability at least $1/\poly(m)$ by our algorithm, implying that we will output $\phi$ (and correctly guess yes) with probability at least $1/\poly(m)$. But no such algorithm for unique-SAT can exist unless NP $\subseteq$ RP. 
\end{proof}

\begin{proof}[Proof of Lemma~\ref{lem:boxs}]
Consider the complete bipartite graph with $m$ item nodes on the left and $m/2$ nodes on the right. Let $v(\cdot)$ denote the binary OXS valuation corresponding to this graph. Then $v(S):= \min\{|S|, m/2\}$, and $v(\cdot)$ is efficiently computable. Observe also that if we pick any set $S$ of size $m/2$ and remove all edges from each left-hand node in $S$ to (say) the top right-hand node (so we remove a total of $m/2$ edges), then we still have a binary OXS function. Moreover, we claim that this function is now $v_S$ for the removed $S$.

Indeed, observe that for any set $T$ of size $<m/2$, there is still a perfect matching from $T$ to the RHS, so we'll still have $v(T) = |T|$. For any set $T$ of size $\geq m/2$, perhaps $T = S$. If so, then clearly the max-weight matching has size $m/2-1$, as desired. Otherwise, $T$ contains an element not in $S$, which has an edge to the top node. There are at least $m/2-1$ left-hand nodes remaining, which have edges to every other right-hand node, and therefore form a matching of size $m/2-1$, for a total matching of $m/2$. 

Therefore, $v(\cdot)$ can be perturbed at any set of size $m/2$. We can index these sets lexicographically to get an efficient mapping from these sets to $[\binom{m}{m/2}]$, as desired.
\end{proof}
\section{Omitted Proofs from Section~\ref{sec:reduction}}\label{app:reduction}
\begin{proof}[Proof of Lemma~\ref{lem:matroid}]
To see the claim, consider the independent sets $\mathcal{I}$ for the matroid defining $v(\cdot)$. It is well-known that updating $\mathcal{I}_y:= \{S \in \mathcal{I}, |S| \leq y\}$ results in $\mathcal{I}_y$ being a matroid.\footnote{This is also easy to check: Clearly $\mathcal{I}_y$ is downwards closed. For any two independent sets $S, T$ in $\mathcal{I}_y$, with $|S| < |T|$, the fact that $\mathcal{I}$ satisfies the augmentation property immediately implies that there exists $i \in T$ such that $S \cup \{i\} \in \mathcal{I}$. As $|T| \leq y$, $S \cup \{i\} \in \mathcal{I}_y$ as well, so $\mathcal{I}_y$ has the augmentation property too.} It is now easy to see that $v(\cdot)$ item-truncated at $y$ is exactly the matroid-based valuation corresponding to $\mathcal{I}_y$ (with the exact same weights $w_e$ for all items $e$), so it is matroid-based as well.
\end{proof}

\begin{proof}[Proof of Lemma~\ref{lem:correctnessitem}]
Throughout the proof, we'll let $M$ denote one copy of the original $m$ items, $M_i$ denote the $i^{th}$ copy, and use variables $S, T,\ldots$ to denote subsets of $\sqcup_i M_i$, and variables $A,B,\ldots$ to denote subsets of $M$.

Let $S = S_1 \sqcup \ldots \sqcup S_{k-1}$ denote the $\alpha$ approximation to the \sadp\ instance $IT(v,w)$. First, we wish to refine $S$ so that $|S| \leq m$, but it is still an $\alpha$-approximation. To do this, observe that if $|S| \leq m$ then we are done. Otherwise, for each $\ell$ there is a subset $T\subseteq S$ with $|T|=m$ and $v_\ell(S) = v_\ell(T)$. To find this $T$, go through each item $i \in S$ one-by-one and check whether $v_\ell(S\setminus\{i\}) =v_\ell(S)$. If so, remove $i$ and continue (because there is still a set $T$ of size $m$ inside $S \setminus \{i\}$ with $v_\ell(T) = v_\ell(S)$). If not, we know that every such $T$ must contain $i$, so keep it. Stop when $|S| = m$. Observe that this procedure must terminate, because any $i \notin T$ will be removed when processed (and that it can be implemented with only value queries to $v_\ell(\cdot)$).

Do the above process for each $\ell$, to get $k$ candidate sets $S^\ell$, each of size at most $m$. Next, simply output $\arg\max_{\ell} \{v(S^\ell_\ell) - w(S^\ell_\ell)\}$ as the solution to the \odp\ instance $(v,w)$. We claim this must be an $\alpha$-approximation.

Consider first the solution $S$, and the index $\ell$ for which $v_\ell(S) - v_{\ell+1}(S) \geq \alpha \cdot \max_{T} \{v_\ell(T) - v_{\ell+1}(T)\}$. First, obesrve that indeed $\max_T \{v_\ell(T) - v_{\ell+1}(T)\} \geq (k+\ell)\max_A \{v(A) - w(A)\}$. This is because one candidate for $T$ is to set $S_\ell = \arg\max_A \{v(A) - w(A)\}$, and $S_i = \emptyset$ for all $i \neq \ell$ (which results in exactly this difference). So we may now conclude that:

$$v_\ell(S) - v_{\ell+1}(S) \geq \alpha(k+\ell) \cdot\max_A \{v(A) - w(A)\}.$$

Next, we claim that we must have $v_\ell(S^\ell) - v_{\ell+1}(S^\ell) \geq v_\ell(S) - v_{\ell+1}(S)$. This is simply because $S^\ell \subseteq S$, but $v_\ell(S^\ell) = v_\ell(S)$ by definition. Therefore, we get:

$$v_\ell(S^\ell) - v_{\ell+1}(S^\ell) \geq \alpha  (k+\ell)\cdot \max_A \{v(A) - w(A)\}.$$

Finally, observe that $|S^\ell| \leq m$, so we know that:
\begin{align*}
v_\ell(S^\ell) - v_{\ell+1}(S^\ell) &=\sum_{i \geq \ell} (k+i) v(S^\ell_i) + \sum_{i < \ell} (k+i) w(S^\ell_i) - \sum_{i > \ell} (k+i)v(S^\ell_i) - \sum_{i \leq \ell} (k+i)w(S^\ell_i)\\
&= (k+\ell)\cdot(v(S^\ell_\ell) - w(S^\ell_\ell))\\
\end{align*}

Therefore, we may conclude further that:

$$v(S^\ell_\ell) - w(S^\ell_\ell) \geq\alpha \cdot  \max_A \{v(A) - w(A)\}.$$

The output of our algorithm is at least as good as $S^\ell_\ell$, because we take the maximum of this over all potential $\ell$, so we get an $\alpha$-approximation.
\end{proof}

\begin{proof}[Proof of Lemma~\ref{lem:balance}]
No matter how large $k$ is, any set with $m$ items can get value from at most $m$ copies of $M$. Within each copy, the maximum possible is $2kw([m])$, so the total value for any set is at msot $2kmw([m])$. Also, for all $\ell$, $\max_S \{v_\ell(S) - v_{\ell+1}(S)\} \geq (k+\ell)$. So the instance is $2mw([m])$-balanced.
\end{proof}

\begin{proof}[Proof of Lemma~\ref{lem:valuequeries}]
The challenge is that, when $|S| \gg m$, it's not clear how to find the best $T\subseteq S$ of size $m$. Indeed, it is not necessarily possible in poly-time when $v,w$ are submodular, and therefore some assumption on valuation classes is necessary. 

But when $v,w$ are matroid-based, recall that each $v_\ell(\cdot)$ is matroid-based as well. This suggests that a greedy algorithm can find the best $T \subseteq S$ of size $m$. Indeed, the following algorithm works: 
\begin{enumerate}
\item For each $e \in S$, find $v_\ell(e)$. This can be done with a single value query to either $v(\cdot)$ or $w(\cdot)$.
\item Sort the elements in $S$ in decreasing order of $v_\ell(e)$. 
\item Initialize $T = \emptyset$. 
\item Process elements $e \in S$ one at a time. 
\begin{enumerate}
\item If $|T| = m$, stop. 
\item Else, compute $v_\ell(T \cup \{e\}) - v_\ell(T)$. This can be done with two queries to either $v(\cdot)$ or $w(\cdot)$ because $|T| < m$.
\item If $v_\ell(T \cup \{e\}) - v_\ell(T) > 0$, add $e$ to $T$. Otherwise, don't. 
\end{enumerate}
\item Output $\sum_{i \geq \ell} (k+i)\cdot v(T_i) + \sum_{i < \ell} (k+i) \cdot w(T_i)$.
\end{enumerate}

The algorithm above finds the max-weight independent subset of $S$, for the matroid defining the matroid-based valuation $v_\ell(\cdot)$, and the correct weights. Therefore, it correctly computes $v_\ell(S)$. 
\end{proof}

\begin{proof}[Proof of Lemma~\ref{lem:compatible}]
Consider the given allocations $(X_1,\ldots,X_k)$ and the multipliers $Q_i:= (C_1+1)^{(i-1)}$. Observe first that all multipliers are integers at most $2^{k \log_2 (C_1)}$, as desired. Also, we are already given that $X_\ell \in \arg\max\{v_\ell(S) - v_{\ell+1}(S)\}$ for all $\ell$. So we just need to establish that $V = (Q_1\cdot v_1,\ldots, Q_k \cdot v_k)$ is compatible with $X = (X_1,\ldots, X_k)$.

First, let's establish that $V$ and $X$ are cyclic monotone. We proceed inductively from $\ell = k$ down to $1$ and prove that $v_\ell$ must be matched to $X_\ell$. Assume for inductive hypothesis that the max-weight matching must match $v_i$ to $X_i$ for all $i > \ell$ and consider now $\ell$. Then consider the weight contributed by the edge between $v_\ell$ and $X_\ell$, versus $v_\ell$ and any $X_i$, $i < \ell$. By property 3, the former edge contributes at least $(C_1+1)^{\ell-1}$ more than any other potential edge. At the same time, by property 2, the maximum possible contribution from \emph{all} edges $<\ell$ is at most $\sum_{i = 1}^{\ell-1} (C_1+1)^{i-1} C_1 \leq C_1 (C_1+1)^{\ell-1}/C_1$. Therefore, the maximum possible weight we can get from edges not adjacent to $v_\ell$ is less than what we lose by not matching $v_\ell$ to $X_\ell$, and the max-weight matching must match $v_\ell$ to $X_\ell$.

A nearly-identical argument establishes that the max-weight matching of valuations $v_{i+1},\ldots, v_j$ to allocations $X_i,\ldots, X_{j-1}$ is to match $X_\ell$ to $v_{\ell+1}$ for all $\ell$. Again proceed inductively from $\ell = j$ down to $i+1$, and assume for inductive hypothesis that the max-weight matching must match $v_a$ to $X_{a-1}$ for all $a > \ell$. Consider again the weight contributed by the edge between $v_\ell$ and $X_{\ell-1}$ versus $v_\ell$ and any $X_i$, $i < \ell-1$. By property 3, the former edge contributes at least $(C_1+1)^{\ell-1}$ more than any other potential edge. At the same time, by property 2, the maximum possible contribution from \emph{all} edges $<\ell$ is at most $\sum_{i = 1}^{\ell-1} (C_1+1)^{i-1} C_1 \leq C_1 (C_1+1)^{\ell-1}/C_1$. Therefore, the maximum possible weight we can get from edges not adjacent to $v_\ell$ is less than what we lose by not matching $v_\ell$ to $X_{\ell-1}$, and the max-weight matching must match $v_\ell$ to $X_{\ell-1}$. 

This completes the proof that $(v_1,\ldots, v_k)$ is $k\log_2(C_1)$-Compatible.
\end{proof}

\begin{proof}[Proof of Corollary~\ref{cor:fulldeal}]
Consider the allocations $X_\ell$ defined so that $X_\ell \cap M_\ell \in \arg\max_A\{v(A) - w(A)\}$, and $X_\ell \cap M_i = \emptyset$ for all $i \neq \ell$. Then we clearly satisfy hypothesis 1 in Lemma~\ref{lem:compatible}. To see this, recall that for any set $S$, $v_\ell(S) - v_{\ell+1}(S) \leq v_\ell(S \cap M_\ell) - v_{\ell+1}(S \cap M_\ell)$, and $X_\ell$ is the optimal set contained in $M_\ell$.

Hypothesis 2 is also clearly satisfied for $C_1 = 2k\max\{v([m]),w([m])\}$. This follows because $v_\ell(M_i) \leq 2k\max\{v([m]),w([m])\}$ for all $\ell$.

Hypothesis 3 is satisfied for the following reasons. First, let $A^*:= \arg\max_A \{v(A)-w(A)\}$. We know that $v(A^*) - w(A^*) \geq 0$. We also know that $v_\ell(X_\ell) = (k+\ell)\cdot v(A^*)$, and $v_\ell(X_i) = (k+i)\cdot w(A^*)$ for all $i < \ell$. As $v(A^*) \geq w(A^*)$, we Hypothesis 3 holds as well.

By Lemma~\ref{lem:compatible}, we may conclude that $IT(v,v_S)$ is $k\log_2(2k\max\{v([m]),w([m])\})$-Compatible.
\end{proof}

\begin{proof}[Proof of Theorem~\ref{thm:mainhardness}]
To see the first claim, assume for contradiction that such an algorithm existed. Then with Theorem~\ref{thm:mainreduction}, we get an algorithm for \odp($\mathcal{V}$) which achieves a non-trivial approximation on all instances $(v,v_S)$ with probability $q$ with a total of $\frac{qx-1}{2} - \poly(km)+ \poly(km)$ value queries (the extra queries come from Lemma~\ref{lem:correctnessitem}). By Lemma~\ref{lem:perturb}, this is not possible. So the original \mdmdp($\mathcal{V}$) algorithm must not exist.

To see the second claim, assume again for contradiction that such an algorithm existed. Then by Theorem~\ref{thm:mainreduction}, this would imply a poly-time algorithm with a non-trivial approximation on all instances $(v,v_S)$ with probability $1/\poly(m)$ with total runtime $\poly(km) = \poly(m)$. By Lemma~\ref{lem:perturb}, this is impossible unless NP $\subseteq$ RP.
\end{proof}

\begin{proof}[Proof of Theorem~\ref{thm:matroidbased}]
This is a corollary of Theorem~\ref{thm:mainhardness} and Lemma~\ref{lem:boxs}. As Binary OXS is $(\text{exp}(m),m/2)$-perturbable, so is matroid-based. Therefore, Theorem~\ref{thm:mainhardness} implies that \mdmdp(matroid-based) cannot be approximated better than $\frac{m^2}{k-1+m^2}$ with probability $1/\poly(m)$ on instances with $(k-1)m$ items and distributions of support $k$ for any $k = \poly(m)$. Letting $k = m^{2/\varepsilon}$, let $m' = km$ denote the number of items. Then there are $k \leq (m^{2/\varepsilon+1})^{1-\varepsilon/3} = (m')^{1-\varepsilon/3}$ valuations in the support, and the best achievable approximation guarantee is $m^{2-2/\varepsilon} \leq (m^{2/\varepsilon+1})^{-1+3\varepsilon/2}$. 

So setting $\varepsilon' = 3\varepsilon/2$, we see that with $m'$ items and $(m')^{1-2\varepsilon'/9}$ valuations in the support of $D$, we can't get better than a $(m')^{1-\varepsilon'}$ approximation.

The only missing step is to connect everything to the demand oracle model. Observe that because all valuations involved are matroid-based (and therefore Gross Substitutes), that in fact a demand oracle can be implemented with $\poly(m)$ value queries. So if we could solve \mdmdp(matroid-based) using demand queries, we could use a value oracle for the $(v,w)$ input to \odp\ to get a value oracle for the \mdmdp\ instance as in the reduction, which is also good enough for demand queries.
\end{proof}

\section{Claims for Value Truncation}\label{app:value}
Below, define the \sadp\ instance $VT(v,w)$ to first use the Disjoint Union reduction, followed by the value-truncation reduction with parameter $2kv([m])$. 

\begin{lemma}\label{lem:correctnessvalues} Let $S$ be an $\alpha$-approximation to the \sadp\ instance $VT(v,w)$. Then with an additional $O(k)$ runtime/value queries (to the value-truncated values), an $\alpha$-approximation to the \odp\ instance $(v,w)$ can be found.
\end{lemma}
\begin{proof}
Throughout the proof, we'll let $M$ denote one copy of the original $m$ items, $M_i$ denote the $i^{th}$ copy, and use variables $S, T,\ldots$ to denote subsets of $\sqcup_i M_i$, and variables $A,B,\ldots$ to denote subsets of $M$.

Let $S = S_1 \sqcup\ldots \sqcup S_{k-1}$ denote the $\alpha$ approximation to the \sadp\ instance $VT(v,w)$ (the resulting \sadp\ instance from the value-truncation reduction). Simply output $\arg\max_{\ell} \{v(S_\ell) - w(S_\ell\}$ (or $\emptyset$ if they all yield a negative answer) as the solution to the \odp\ instance $(v,w)$. We claim this must be an $\alpha$-approximation.

Consider first the solution $S$, and the index $\ell$ for which $v_\ell(S) - v_{\ell+1}(S) \geq \alpha \cdot \max_{T} \{v_\ell(T) - v_{\ell+1}(T)\}$. First, observe that indeed $\max_T \{v_\ell(T) - v_{\ell+1}(T)\} \geq (k+\ell)\max_A \{v(A) - w(A)\}$. This is because one candidate for $T$ is to set $S_\ell = \arg\max_A \{v(A) - w(A)\}$, and $S_i = \emptyset$ for all $i \neq \ell$ (which results in at least this difference, because our truncation exceeds $2kv([m])$). So we may now conclude that:

$$v_\ell(S) - v_{\ell+1}(S) \geq \alpha(k+\ell) \cdot\max_A \{v(A) - w(A)\}.$$

Next, there are a few cases to consider. First, perhaps $v_{\ell+1}(S) = 2kv([m])$. In this case, $v_\ell(S) \leq 2kv([m])$, so the LHS is $0$, which guarantees an $\alpha$-approximation. Therefore, $\emptyset$ also guarantees an $\alpha$-approximation, so our algorithm succeeds. Next, perhaps $v_{\ell+1}(S) < 2kv([m])$. In this case, we get that:

\begin{align*}
v_\ell(S) - v_{\ell+1}(S) &= v_\ell(S) - \sum_{i > \ell} (k+i)v(S_i) - \sum_{i \leq \ell} (k+i)w(S_i)\\
& \leq \sum_{i \geq \ell} (k+i) v(S_i) + \sum_{i < \ell} (k+i)w(S_i) - \sum_{i > \ell} (k+i)v(S_i) - \sum_{i \leq \ell} (k+i) w(S_i)\\
&=(k+\ell)\cdot(v(S_\ell) - w(S_\ell))\\
\end{align*}

Therefore, we may conclude further that:

$$v(S_\ell) - w(S_\ell) \geq\alpha \cdot  \max_A \{v(A) - w(A)\}.$$

The output of our algorithm is at least as good as $S_\ell$, because we take the maximum of this over all potential $\ell$, so we get an $\alpha$-approximation.
\end{proof}

\begin{lemma}\label{lem:balancevalue} For any \odp\ instance $(v,w)$, the instance $VT(v,w)$ is $2v([m])$-balanced.
\end{lemma}
\begin{proof}
Value-truncation explicitly caps $v_k(M_1 \sqcup \ldots \sqcup M_{k-1})$ at $2kv([m])$. Also, for all $\ell$, $\max_S \{v_\ell(S) - v_{\ell+1}(S)\} \geq (k+\ell)$. So the instance is $2v([m])$-balanced.
\end{proof}

\begin{lemma}\label{lem:valuequeriesvalue} All $v,w$ are query-sufficient for $VT(v,w)$.
\end{lemma}
\begin{proof}
Simply observe that computing a value query for a function $v_\ell(S_1\sqcup\ldots \sqcup S_{k-1})$ simply requires computing either $v(S_i)$ or $w(S_i)$ for all $i$, summing them, and taking a minimum with $x$, all of which can be done in time $O(k)$.
\end{proof}

\section{Example: Submodular is not closed under Item Truncation}
Here we provide an example of a submodular function $v$ and size parameter $y$ such that $v':= \max_{T \subseteq S, |T| \leq y} \{v(T)\}$. This is sufficient to prove that submodular functions are not closed under item truncation, as our transformation of $v$ is a special instance of item truncation with $y = 2$.

\begin{table}[H]
\centering
\caption{Valuation function $v$}
\begin{tabular}{ccccccccl}
Sets of size 1 & \multicolumn{1}{c|}{Value} & Sets of size 2 & \multicolumn{1}{c|}{Value} & Sets of size 3 & \multicolumn{1}{c|}{Value} & Sets of size 4 & Value &  \\ \cline{1-8}
\{a\}          & \multicolumn{1}{c|}{5}     & \{c,d\}        & \multicolumn{1}{c|}{10}    & \{a,b,c\}      & \multicolumn{1}{c|}{9}     & \{a,b,c,d\}    & 10    &  \\
\{b\}          & \multicolumn{1}{c|}{5}     & \{a,d\}        & \multicolumn{1}{c|}{9}     & \{a,b,d\}      & \multicolumn{1}{c|}{10}    &                &       &  \\
\{c\}          & \multicolumn{1}{c|}{5}     & \{b,d\}        & \multicolumn{1}{c|}{9}     & \{a,c,d\}      & \multicolumn{1}{c|}{10}    &                &       &  \\
\{d\}          & \multicolumn{1}{c|}{5}     & \{a,c\}        & \multicolumn{1}{c|}{9}     & \{b,c,d\}      & \multicolumn{1}{c|}{10}    &                &       &  \\
               & \multicolumn{1}{c|}{}      & \{b,c\}        & \multicolumn{1}{c|}{9}     &                & \multicolumn{1}{c|}{}      &                &       &  \\
               & \multicolumn{1}{c|}{}      & \{a,b\}        & \multicolumn{1}{c|}{9}     &                & \multicolumn{1}{c|}{}      &                &       &  \\
               &                            &                &                            &                &                            &                &       & 
\end{tabular}
\end{table}

We will use the following definition for submodularity: for any set of items $X$ and any items $y,z \notin X$, $v(X\cup y \cup z) - v(X \cup y) \leq v(X \cup z) - v(X)$. We will prove this by considering all possible sizes for $X$. If $X = \emptyset$: for all items $y$ and $z$, $v(y \cup z) - v(y) \leq 5 = v(z) - v(\emptyset)$. If $|X| = 1$: for all items $y$ and $z$, $v(X \cup y \cup z) - v(X \cup y) \leq 1 < 4 \leq v(X \cup z) - v(X)$. Now consider the case where $|X| = 2$. If $z \ne d$, $X \cup y$ cannot be the set $\{a,b,c\}$. In this case, $v(X \cup y) = v(X \cup z) = 10$, and therefore $v(X \cup y \cup z) - v(X \cup y) = 0 \leq v(X \cup z) - v(X)$. If $z = d$, then $X \ne \{c,d\}$, and therefore $v(X) = 9$. $X \cup y$ must be equal to $\{a,b,c\}$ and $X \cup z \ne \{a,b,c\}$, so $v(X \cup y) = 9$ and $v(X \cup z) = 10$. Putting it all together, $v(X\cup y \cup z) - v(X \cup y) = 10 - 9 = 1 = v(X \cup z) - v(X)$.
\\
\\
The function is also monotone. Note that the only sets $(X, Y)$ where $|X| < |Y|$ but $v(X) > v(Y)$ are $(\{c,d\},\{a,b,c\})$. In this case $X \nsubseteq Y$, so it does not violate monotonicity. Thus, $v$ represents a monotone submodular function (so there's nothing "weird" about our counterexample). 
\\
\\
Now let us define a function $v'$ as $\max_{T \subseteq S, |T| \leq 2} \{v(T)\}$. Note that $v$ and $v'$ will differ in evaluation only over $\{a,b,d\}$, where $v(a,b,d) = 10$ and $v'(a,b,d) = 9$. All other subsets of size greater than $2$ contain the subset $\{c,d\}$ with value $10$. 
\\
\\
$v'$ is not submodular. Note that $$v'(a,b,c,d) - v'(a,b,d) = 1 > 0 = v'(a,b,c) - v'(a,b)$$
If we let $X = \{a,b\}$, $x = c$ and $y = d$, we see that
$v'(X \cup x \cup y) - v'(X \cup y) > v'(X \cup x) - v'(X)$,
thus contradicting the submodularity definition. 

\section{Valuation Classes}\label{app:valuations}
The following valuation classes are referenced throughout the text. They are listed below in strictly increasing order of generality.
\begin{itemize}
\item Additive: for all $S$, $v(S):=\sum_{i \in S} v(\{i\})$. Unit-demand: for all $S$, $v(S) := \max_{i \in S} v(\{i\})$.
\item $c$-demand: for all $S$, $v(S) := \max_{T \subseteq S, |T| \leq c} \{\sum_{i \in T} v(\{i\})$.
\item OXS: there is a weighted bipartite graph $G$ with nodes $L= [m]$ on the left and $R$ on the right. $v(S)$ is the size of the max weight matching using nodes $S$ on the left and $R$ on the right. When all weights are $1$ or $0$, call this binary OXS.
\item Matroid-based: let $\mathcal{I}$ be independent sets of a matroid on $[m]$, and $w_i$ be weights for each $i\in [m]$. Then $v(S):=\max_{T \subseteq S, T \in \mathcal{I}}\{\sum_{i \in T} w_i\}$. When each $w_i = 1$, call this matroid-rank.
\item Gross Substitutes: Let $\vec{q},\vec{p}$ be price vectors with $q_i \geq p_i$ for all $i$. Let $P:=\{i, q_i = p_i\}$. For all $X \in \arg\max_S\{v(S) - \sum_{i \in S} p_i\}$, there exists $Y \in \arg\max_S\{v(S) - \sum_{i \in S} q_i\}$ with $(P \cap X) \subseteq Y$.
\item Submodular: for all $S, T$, $v(S) + v(T) \geq v(S \cap T)+v(S \cup T)$.
\item XOS:\footnote{Equivalently, Fractionally Subadditive.} There exist additive functions $w_1,\ldots$ such that for all $S$, $v(S):= \max_j \{w_j(S)\}$.
\item Subadditive: for all $S, T$, $v(S) + v(T) \geq v(S \cup T)$.
\end{itemize}

\end{document}